\documentclass[11pt]{llncs}

\usepackage[latin5]{inputenc}

\usepackage{amssymb}
\usepackage{amsmath}
\usepackage{amsthm}
\usepackage{textcomp}
\usepackage{array}
\usepackage{multirow}
\usepackage{color}

\newcommand{\bra}[1]{\langle #1|}
\newcommand{\ket}[1]{|#1\rangle}

\newcommand{\cent}[0]{\mbox{\textcent}}
\newcommand{\dollar}[0]{\$}
\title{Probabilistic and quantum finite automata with postselection\thanks{This 
work was partially supported by the Scientific and Technological
Research Council of Turkey (T\"{U}B\.ITAK) with grant 108E142.}{\small'}\thanks{A preliminary version of this paper appeared in the \textit{Proceedings of Randomized and Quantum Computation (satellite workshop of MFCS and CSL 2010)}, pages 14--24, 2010.}}
\author{Abuzer Yakary{\i}lmaz\ \and A. C. Cem Say }
\institute{Bo\u{g}azi\c{c}i University, Department of Computer Engineering,\\ Bebek 34342 \.{I}stanbul, Turkey \\
\email{{abuzer,say}@boun.edu.tr}
 \\~~\\
\today
}

\begin{document}

\newlength{\twidth}
\maketitle
\pagenumbering{arabic}

%-----------------------------------------------------------------------------%
\begin{abstract} \label{abstract:Abstract}
%-----------------------------------------------------------------------------%

We prove that endowing a real-time probabilistic or quantum computer with the ability of postselection increases its computational power. For this purpose, we provide a new model of 
finite automata with postselection, and compare it with the model of L\={a}ce et al. We examine the related language classes, and also establish separations between the classical and quantum versions, and between the zero-error vs. bounded-error modes of recognition in this model.

\end{abstract}

% SSSSSSSSSSSSSSSSSSSSSSSSSSSSSSSSSSSSSSSSSSSSSSSSSSSSSSSSSSSSSSSSSSSSSSSSSSSSSSSS %
% SSSSSSSSSSSSSSSSSSSSSSSSSSSSSSSSSSSSSSSSSSSSSSSSSSSSSSSSSSSSSSSSSSSSSSSSSSSSSSSS %
% SSSSSSSSSSSSSSSSSSSSSSSSSSSSSSSSSSSSSSSSSSSSSSSSSSSSSSSSSSSSSSSSSSSSSSSSSSSSSSSS %
\section{Introduction} \label{section:Introduction}
% SSSSSSSSSSSSSSSSSSSSSSSSSSSSSSSSSSSSSSSSSSSSSSSSSSSSSSSSSSSSSSSSSSSSSSSSSSSSSSSS %
% SSSSSSSSSSSSSSSSSSSSSSSSSSSSSSSSSSSSSSSSSSSSSSSSSSSSSSSSSSSSSSSSSSSSSSSSSSSSSSSS %
% SSSSSSSSSSSSSSSSSSSSSSSSSSSSSSSSSSSSSSSSSSSSSSSSSSSSSSSSSSSSSSSSSSSSSSSSSSSSSSSS %

The notion of postselection as a mode of computation was introduced by
Aaronson \cite{Aa05}. Postselection is the (unrealistic) capability
of discarding all branches of a computation in which a specific event
does not occur, and focusing on the surviving branches for the final
decision about the membership of the input string in the recognized
language. Aaronson examined $\mathsf{PostBQP}$, the class of languages recognized with bounded error
by polynomial-time quantum computers with postselection, and showed it
to be identical to the well-known classical complexity class $\mathsf{PP}$. It is, however, still an open question whether postselection adds anything to the power of polynomial-time computation, since we do not even know whether $\mathsf{P}$, the class of languages recognized by classical computers with zero error in polynomial time, equals $\mathsf{PP}$ or not. In this paper, we prove that postselection \textit{is} useful for real-time computers with a constant space bound, that is, finite automata.

Groundbreaking work on the effect of postselection on quantum finite
automata (QFAs) was performed by L\={a}ce, Scegulnaja-Dubrovska, and
Freivalds \cite{LSF09}, who defined a model that is somewhat different (and, as we show here, strictly more powerful,)
than Aaronson's basic concept. In this paper, we examine the power of
postselection on both probabilistic and quantum finite automata. Our
model of postselection is more in alignment with Aaronson's original
definition. We establish some basic properties of the related language
classes and the relationships among them. It turns out that classical
probabilistic finite automata (PFAs) with
(our kind of) postselection are strictly more powerful than ordinary PFAs, and that QFAs with postselection are even more powerful than their classical counterparts. We also prove that QFAs with postselection have the same computational power as the
recently introduced real-time QFAs with restart \cite{YS10B}, and allowing a small but positive error to be committed by a finite automaton with postselection enlarges the class of recognized languages in comparison to the zero-error case.

% SSSSSSSSSSSSSSSSSSSSSSSSSSSSSSSSSSSSSSSSSSSSSSSSSSSSSSSSSSSSSSSSSSSSSSSSSSSSSSSS %
% SSSSSSSSSSSSSSSSSSSSSSSSSSSSSSSSSSSSSSSSSSSSSSSSSSSSSSSSSSSSSSSSSSSSSSSSSSSSSSSS %
% SSSSSSSSSSSSSSSSSSSSSSSSSSSSSSSSSSSSSSSSSSSSSSSSSSSSSSSSSSSSSSSSSSSSSSSSSSSSSSSS %
\section{Standard models of probabilistic and quantum finite automata} \label{sec:PFA-QFA}
% SSSSSSSSSSSSSSSSSSSSSSSSSSSSSSSSSSSSSSSSSSSSSSSSSSSSSSSSSSSSSSSSSSSSSSSSSSSSSSSS %
% SSSSSSSSSSSSSSSSSSSSSSSSSSSSSSSSSSSSSSSSSSSSSSSSSSSSSSSSSSSSSSSSSSSSSSSSSSSSSSSS %
% SSSSSSSSSSSSSSSSSSSSSSSSSSSSSSSSSSSSSSSSSSSSSSSSSSSSSSSSSSSSSSSSSSSSSSSSSSSSSSSS %

% sssssssssssssssssssssssssssssssssssssssssssssssssssssssssssssssssssssssssssssssss %
% sssssssssssssssssssssssssssssssssssssssssssssssssssssssssssssssssssssssssssssssss %
\subsection{Probabilistic finite automata} \label{sec:PFA}
% sssssssssssssssssssssssssssssssssssssssssssssssssssssssssssssssssssssssssssssssss %
% sssssssssssssssssssssssssssssssssssssssssssssssssssssssssssssssssssssssssssssssss %

A real-time probabilistic finite automaton (RT-PFA) is a 5-tuple
\begin{equation}
      \mathcal{P}=(Q,\Sigma,\{ A_{\sigma \in \tilde{\Sigma}} \},q_{1},Q_{a}),
\end{equation}
where $ Q $ is the set of internal states,    $ q_{1} $ is the
initial state, $ Q_{a} \subseteq Q
$ is the set
of accepting states, $ \Sigma $ is the input alphabet, not containing the end-markers
$ \cent $ and $ \dollar $, $ \tilde{\Sigma} = \Sigma \cup \{ \cent, \dollar \} $, and the $ A_{\sigma} $ are  transition matrices, whose columns are
stochastic vectors, such that $ A_{\sigma} $'s $ (j,i)^{th}
$ entry, denoted  $ A_{\sigma}[j,i] $, is the probability of the transition from state $
q_{i} $ to state $ q_{j} $
when reading symbol $ \sigma $.

The computation of a RT-PFA can be traced by a stochastic state vector,
say $ v $, whose $ i^{th} $ entry, denoted
$ v[i] $, corresponds to state $ q_{i} $.
For a given input string $ w \in \Sigma^{*} $ (the string read by the machine is $ \tilde{w} = \cent w \dollar $),
\begin{equation}
\label{equation:vAv}
      v_{i} = A_{\tilde{w}_{i}} v_{i-1},
\end{equation}
where $ \tilde{w}_{i} $ denotes the $ i^{th} $ symbol of $ \tilde{w} $,  $ 1 \le i \le | \tilde{w} | $, and
$ v_{0} $ is the initial state vector, whose first entry is 1. ($ |\tilde{w}| $  denotes the length of $ \tilde{w} $.)
The transition matrices of a RT-PFA can be extended for any string as
\begin{equation}
      A_{w\sigma} = A_{\sigma} A_{w},
\end{equation}
where $ w \in (\tilde{\Sigma})^{*} $, $ \sigma \in \tilde{\Sigma} $, and
$ A_{\varepsilon} = I $ ($ \varepsilon $ denotes the empty string).
The  probability that RT-PFA $ \mathcal{P} $ will accept string $ w $ is
\begin{equation}
      f_{\mathcal{P}}^{a}(w) = \sum_{q_{i} \in Q_{a}} (A_{\tilde{w}}v_{0})[i] =
              \sum_{q_{i} \in Q_{a}} v_{|\tilde{w}|}[i].
\end{equation}
The probability that $ \mathcal{P} $ will reject string $ w $ is $f_{\mathcal{P}}^{r}(w)=1-f_{\mathcal{P}}^{a}(w)$.

The language $ L \subseteq \Sigma^{*} $ recognized by machine $
\mathcal{M} $ with (strict) cutpoint
$ \lambda \in \mathbb{R} $ is defined as
\begin{equation}
       L = \{ w \in \Sigma^{*} \mid f_{\mathcal{M}}^{a}(w) > \lambda \}.
\end{equation}
The languages recognized by RT-PFAs 
 with cutpoint
 form the class of \textit{stochastic languages}, denoted $ \mathsf{S} $.

The language $ L \subseteq \Sigma^{*} $ recognized by machine $
\mathcal{M} $ with nonstrict cutpoint
$ \lambda \in \mathbb{R} $ is defined as \cite{BJKP05}
\begin{equation}
       L = \{ w \in \Sigma^{*} \mid f_{\mathcal{M}}^{a}(w) \geq \lambda \}.
\end{equation}
The languages recognized by RT-PFAs 
 with nonstrict cutpoint form the class of \textit{co-stochastic languages}, denoted $ \mathsf{coS} $.

$ \mathsf{S} $ $ \cup $ $ \mathsf{coS} $ (denoted $ \mathsf{uS} $) 
is the class of languages recognized by RT-PFAs
with unbounded error.
 
Probabilistic automata that recognize a language with cutpoint zero are identical to nondeterministic automata, in particular, the class of languages recognized by RT-PFAs with cutpoint zero is $ \mathsf{REG} $ \cite{Bu67}, the class of regular languages.

The language $ L \subset \Sigma^{*} $ is said to be recognized by machine $
\mathcal{M} $ with error bound $ \epsilon $
($ 0 \le \epsilon < \frac{1}{2} $) if
\begin{itemize}
       \item $ f_{\mathcal{M}}^{a}(w) \ge 1 - \epsilon $ for all $ w \in L $, and,
       \item $ f_{\mathcal{M}}^{r}(w) \ge 1 - \epsilon $ for all $ w \notin L $.
\end{itemize}
This situation is also known as recognition with bounded error.

RT-PFAs recognize precisely the regular languages with bounded error \cite{Ra63}.

Viewing the input as written (between the end-markers) on a suitably long tape, with each tape square containing one symbol from $ \tilde{\Sigma} $, and a tape head moving over the tape, sending the symbol it currently senses to the machine for processing, the RT-PFA model can be augmented by allowing the transition matrices to specify the direction in which the tape head can move in each step, as well as the next state.  The model obtained by legalizing leftward and stationary tape head moves in this manner is named the \textit{two-way probabilistic finite automaton} (2PFA). 2PFAs can recognize some nonregular languages with bounded error in exponential time \cite{Fr81}.

% sssssssssssssssssssssssssssssssssssssssssssssssssssssssssssssssssssssssssssssssss %
% sssssssssssssssssssssssssssssssssssssssssssssssssssssssssssssssssssssssssssssssss %
\subsection{Quantum finite automata} \label{sec:QFA}
% sssssssssssssssssssssssssssssssssssssssssssssssssssssssssssssssssssssssssssssssss %
% sssssssssssssssssssssssssssssssssssssssssssssssssssssssssssssssssssssssssssssssss %

A real-time quantum finite automaton (RT-QFA) \cite{Hi08,Ya11A,YS11A} is a 5-tuple
\begin{equation}
      \mathcal{M}=(Q,\Sigma,\{\mathcal{E}_{\sigma \in
\tilde{\Sigma}}\},q_{1},Q_{a}),
\end{equation}
where $ Q$, $\Sigma$, $q_{1}$, and $Q_{a} $ are as defined above for RT-PFAs, and $ \mathcal{E}_{\sigma } $ is an admissible operator having the elements
$ \{ E_{\sigma,1},\ldots,E_{\sigma,k} \} $ for some $ k \in \mathbb{Z}^{+} $
satisfying
\begin{equation}
      \sum_{i=1}^{k} E_{\sigma,i}^{\dagger} E_{\sigma,i} = I.
\end{equation}
Additionally, we define the projector
\begin{equation}
      P_{a} = \sum_{q \in Q_{a}} \ket{q}\bra{q}
\end{equation}
in order to check for acceptance.
For a given input string $ w \in \Sigma^{*} $ (the string read by the machine is $ \tilde{w} = \cent w \dollar $), the overall state of the machine can be traced by
\begin{equation}
      \rho_{j} = \mathcal{E}_{\tilde{w}_{j}} (\rho_{j-1}) =
      \sum_{i=1}^{k} E_{\tilde{w}_{j},i} \rho_{j-1}
E_{\tilde{w}_{j},i}^{\dagger},
\end{equation}
where $ 1 \le j \le | \tilde{w} |  $ and $ \rho_{0} = \ket{q_{1}}
\bra{q_{1}} $ is the initial density matrix.
The transition operators can  be extended easily for any string as
\begin{equation}
      \mathcal{E}_{w \sigma} = \mathcal{E}_{\sigma} \circ \mathcal{E}_{w},
\end{equation}
where $ w \in (\tilde{\Sigma})^{*} $, $ \sigma \in \tilde{\Sigma} $,
and $ \mathcal{E}_{\varepsilon} = I $.
(Note that
$ \mathcal{E}^{\prime} \circ \mathcal{E} $ is described by the collection  $ \{ E^{\prime}_{j} E_{i}
\mid 1 \le i \le k, 1 \le j \le k^{\prime} \} $,
when $ \mathcal{E} $ and $ \mathcal{E}^{\prime} $ are described by the collections 
$ \{E_{i} \mid 1 \le i \le k\} $ and $ \{E_{j}^{\prime} \mid 1 \le j \le k^{\prime}\} $, respectively.)
The  probability that  RT-QFA $ \mathcal{M} $ will accept input string $ w $ is
\begin{equation}
      f_{\mathcal{M}}^{a}(w) = tr( P_{a}
\mathcal{E}_{\tilde{w}}(\rho_{0})) = tr(P_{a} \rho_{| \tilde{w} |} ).
\end{equation}
The class of languages recognized by RT-QFAs with cutpoint (respectively, nonstrict cutpoint)
are denoted $ \mathsf{QAL} $ (respectively, $ \mathsf{coQAL} $).
$ \mathsf{QAL} $ $ \cup $ $ \mathsf{coQAL} $, denoted $ \mathsf{uQAL} $, 
is the class of languages recognized by RT-QFAs with unbounded error.
Any quantum automaton with a sufficiently general definition can simulate its probabilistic counterpart, so one 
direction of the relationships that we report among probabilistic and quantum language classes is always easy to 
see. It is known that $ \mathsf{S} $ = $ \mathsf{QAL} $ , $ \mathsf{coS} $ = $ \mathsf{coQAL} $, 
and $ \mathsf{uS} $ = $ \mathsf{uQAL} $ \cite{Ya11A,YS11A}.
The class of languages recognized by RT-QFAs with cutpoint zero, denoted
$ \mathsf{NQAL} $, is a proper superclass of $ \mathsf{REG} $, and is not closed under complementation \cite{YS10A}.
The class of languages whose complements are in $ \mathsf{NQAL} $ is denoted $ \mathsf{coNQAL} $.
RT-QFAs recognize precisely the regular languages with bounded error \cite{KW97,Bo03,Je07,AY11A}.

% sssssssssssssssssssssssssssssssssssssssssssssssssssssssssssssssssssssssssssssssss %
% sssssssssssssssssssssssssssssssssssssssssssssssssssssssssssssssssssssssssssssssss %
\subsection{Probabilistic and quantum finite automata with restart} \label{sec:restart}
% sssssssssssssssssssssssssssssssssssssssssssssssssssssssssssssssssssssssssssssssss %
% sssssssssssssssssssssssssssssssssssssssssssssssssssssssssssssssssssssssssssssssss %

In this subsection, we review models of  finite automata with restart (see \cite{YS10B} for details).
Since these are two-way machines, the input is written on a tape scanned by a two-way
tape head. For a given input string $ w \in \Sigma $, $ \tilde{w} $ is written on tape, the tape squares are indexed by integers, and $ \tilde{w} $
is written on the squares indexed $ 1 $ through $ |\tilde{w}| $.
For these machines, we assume that after reading the right end-marker $ \dollar $, the input head never tries to visit 
the square indexed by $ |\tilde{w}|+1 $.

A real-time probabilistic finite automaton with restart (RT-PFA$ ^{\circlearrowleft}  $), can be seen as an augmented RT-PFA, and  a 7-tuple
\begin{equation}
      \mathcal{P}=(Q,\Sigma,\{ A_{\sigma \in \tilde{\Sigma}} \},q_{1},Q_{a},Q_{r},Q_{\circlearrowleft} ),
\end{equation}
where $ Q_{r} $ is the set of reject states, and
$ Q_{\circlearrowleft}  $ is the set of restart states.
Moreover, $ Q_{n} = Q \setminus ( Q_{a} \cup Q_{r} \cup Q_{\circlearrowleft} ) $  is the set of
nonhalting and nonrestarting states.
The processing of input symbols by a RT-PFA$ ^{\circlearrowleft}  $ is performed according to Equation \ref{equation:vAv}, as in the RT-PFA, with the additional feature that
after each transition, the internal state $q$ is checked, with the following consequences:
\begin{itemize}
	\item ($  ``\circlearrowleft" $) if $ q \in Q_{\circlearrowleft} $, the computation is restarted 
		(the internal state is set to $ q_{1} $ and the input head is replaced to the square indexed by $ 1 $);
	\item ($ ``a" $) if $ q \in Q_{a} $, the computation is terminated with the decision of acceptance;
	\item ($ ``r" $) if $ q \in Q_{r} $, the computation is terminated with the decision of rejection;
	\item ($ ``n" $) if $ q \in Q_{n} $, the input head is moved one square to the right.
\end{itemize}

The quantum counterpart of RT-PFAs presented in \cite{YS10B} has a parallel definition, but with a 
real-time Kondacs-Watrous quantum finite automaton\footnote{
	We refer the reader to \cite{KW97,Ya11A,YS11A} for details of this QFA variant.
} (RT-KWQFA), rather than a RT-QFA taken as basis. 
A real-time (Kondacs-Watrous) quantum finite automaton with restart (RT-KWQFA$ ^{\circlearrowleft}  $) is a 7-tuple
\begin{equation}
      \mathcal{M}=(Q,\Sigma,\{ U_{\sigma \in \tilde{\Sigma}} \},q_{1},Q_{a},Q_{r},Q_{\circlearrowleft} ),
\end{equation}
where
$ \{ U_{\sigma \in \tilde{\Sigma}} \} $ is a set of unitary transition matrices defined 
for each $ \sigma \in \tilde{\Sigma} $.
The computation of $ \mathcal{M} $ starts with $ \ket{q_{1}} $. 
At each step of the computation, the transition matrix corresponding to the current input symbol, say $ U_{\sigma} $, is applied on the current state vector, 
say $ \ket{\psi} $,
belonging to the state space $ \mathcal{H}_{Q} $, spanned by $ \ket{q_{1}}, \ldots, \ket{q_{|Q|}} $,
and then we obtain new state vector $ \ket{\psi'} $, i.e.
\begin{equation}
	\ket{\psi'} = U_{\sigma} \ket{\psi}.
\end{equation}
After that, a projective  measurement
\begin{equation}
	P = \{ P_{\tau \in \Delta} \mid P_{\tau} = \sum_{q \in Q_{\tau}} \ket{q}\bra{q} \}.
\end{equation}
with outcomes 
$ \Delta=\{``\circlearrowleft",``a",``r",``n"\} $ is performed on the state space.
After the measurement, the machine acts according to the measurement outcomes,
as listed above for RT-PFA$ ^{\circlearrowleft} $s.
Note that the new state vector is normalized after the measurement, and the state vector is set to $ \ket{q_{1}} $ when the computation 
is restarted .

A segment of computation of an automaton with restart $ \mathcal{A} $ which begins with a (re)start, and ends with
a halting or restarting state will be called a \textit{round}. Let $ p^{a}_{\mathcal{A}} (w) $ ($ p^{r}_{\mathcal{A}} (w) $) be the probability that $w$ is accepted (rejected) in a single round of $ \mathcal{A} $. For a given input string $ w \in \Sigma^{*} $,
the overall acceptance and rejection probabilities of $w$ ($ f_{\mathcal{A}}^{a}(w)$ and  $f_{\mathcal{A}}^{r}(w)$, respectively,) can be calculated as shown in the following lemma.

\begin{lemma}
	\label{lem:overall-acc-rej}
	$ f_{\mathcal{A}}^{a}(w)=\frac{ p_{\mathcal{A}}^{a} (w) }{ p_{\mathcal{A}}^{a} (w) + p_{\mathcal{A}}^{r} (w) } $
	and 
	$ f_{\mathcal{A}}^{r}(w)=\frac{ p_{\mathcal{A}}^{r} (w) }{ p_{\mathcal{A}}^{a} (w) + p_{\mathcal{A}}^{r} (w) } $.
\end{lemma}
\begin{proof}
	\begin{eqnarray*}
		f_{\mathcal{A}}^{a}(w) & = &
        \sum_{i=0}^{\infty}\left(1-p_{\mathcal{A}}^{a}(w)-p_{\mathcal{A}}^{r}(w) \right)^{i}
        p_{\mathcal{A}}^{a}(w)\\
		& = & p_{\mathcal{A}}^{a}(w) \left(
		\dfrac{1}{1-(1-p_{\mathcal{A}}^{a}(w)-p_{\mathcal{A}}^{r}(w))} \right) \\
		& = &
		\dfrac{p_{\mathcal{A}}^{a}(w)}{p_{\mathcal{A}}^{a}(w)+p_{\mathcal{A}}^{r}(w)}
	\end{eqnarray*}
	$ f_{\mathcal{A}}^{r}(w) $ is calculated in the same way.
\end{proof}

Moreover, if $ \mathcal{A} $ recognizes a language with error bound $ \epsilon < \frac{1}{2} $,
we have the following relation.
\begin{lemma}
	\label{lem:bounded-error}
    The language $ L \subseteq \Sigma^{*} $ is recognized by $ \mathcal{A} $ with error bound $ \epsilon $
	if and only if $ \frac{ p_{\mathcal{A}}^{r}(w) }{ p_{\mathcal{A}}^{a}(w) } 
	\le \frac{\epsilon}{1-\epsilon} $ when $ w \in L $, 
	and $ \frac{ p_{\mathcal{A}}^{a}(w) }{ p_{\mathcal{A}}^{r}(w) } \le \frac{\epsilon}{1-\epsilon} $
	when $ w \notin L $.
	Furthermore, if $ \frac{p_{\mathcal{A}}^{r}(w)}{p_{\mathcal{A}}^{a}(w)} $ 
	$ \left( \frac{p_{\mathcal{A}}^{a}(w)}{p_{\mathcal{A}}^{r}(w)} \right) $ is at most $ \epsilon $, then
	$ f_{\mathcal{A}}^{a}(w) $ $ ( f_{\mathcal{A}}^{r}(w) ) $ is at least $ 1-\epsilon $.
\end{lemma}
\begin{proof}
       See \cite{YS10B}.
\end{proof}

\begin{lemma}
 \label{lemma:expected-runtime}
 Let $ p_{\mathcal{A}}(w)= p_{\mathcal{A}}^{a}(w) + p_{\mathcal{A}}^{r}(w) $, and let $ s_{\mathcal{A}}(w) $ be the maximum number
of steps in any branch of a
 round of $ \mathcal{A} $ on $ w $.
 The worst-case expected runtime of $ \mathcal{A} $ on $ w $ is
 \begin{equation}
       \label{equation:expected-runtime}
       \frac{1}{p_{\mathcal{A}}(w)} (s_{\mathcal{A}}(w)).
 \end{equation}
\end{lemma}
\begin{proof}
 The worst-case expected running time of $ \mathcal{A} $ on $ w $ is
    \begin{equation}
            \label{equation:expected-runtime-proof}
            \begin{array}{ll}
            	\multicolumn{2}{l}{\sum_{i=0}^{\infty} 
            		(i+1)(1-p_{\mathcal{A}}(w))^{i} (p_{\mathcal{A}}(w))(s_{\mathcal{A}}(w))} \\
            		~~ & = (p_{\mathcal{A}}(w))(s_{\mathcal{A}}(w))\frac{1}{p_{\mathcal{A}}(w)^{2}} \\
            		& =\frac{1}{p_{\mathcal{A}}(w)}(s_{\mathcal{A}}(w)).
            \end{array}
    \end{equation}
\end{proof}

In this paper, we will find it useful to use automata with restart that employ the restart move  only when the input head is at the right end of the input tape. It is obvious that the computational power of RT-PFA$ ^{\circlearrowleft} $s does not change if the act of entering the halting states is postponed to the end of the computation. For the quantum version, it is more convenient to use the general QFA model described in Section \ref{sec:PFA-QFA}, rather than the KWQFA, as the building block of the RT-QFA$^{\circlearrowleft} $ model for this purpose. 
We use the denotation RT-QFA$ ^{\circlearrowleft}$ for this variant of quantum automata with restart.

		A real-time \textit{general} quantum finite automaton with restart (RT-QFA$ ^{\circlearrowleft} $) 
		is a 6-tuple
		\begin{equation}
			\mathcal{M} = (Q,\Sigma,\{\mathcal{E}_{\sigma \in \tilde{\Sigma}}\},q_{1},Q_{a},Q_{r}),
		\end{equation} 
		where all specifications are the same as RT-QFA (see Section \ref{sec:QFA}),
		except that
		\begin{itemize}
			\item $ Q_{r} $ is the set of rejecting states;
			\item $ Q_{\circlearrowleft} = Q \setminus ( Q_{a} \cup Q_{r} ) $ is the set of  restart states;
			\item $ \Delta = \{a,r,\circlearrowleft\} $ with the following specifications:
				\begin{itemize}
					\item ''a": halt and accept,
					\item ''r": halt and reject, and
					\item ''$ \circlearrowleft $":  restart the computation.
				\end{itemize}
				The corresponding projectors, $ P_{a}, P_{r}, P_{\circlearrowleft} $, 
				are defined in a standard way, based on the related sets of states, 
				$ Q_{a}, Q_{r}, Q_{\circlearrowleft} $, respectively.
		\end{itemize}

	Note that a RT-KWQFA$ ^{\circlearrowleft}  $ can be simulated by a RT-QFA$ ^{\circlearrowleft} $
	in a straightforward way, by postponing each intermediate measurement to the end of the computation.
	Formally, for a given RT-KWQFA$ ^{\circlearrowleft}  $ 
	$ \mathcal{M} = ( Q,\Sigma,\{U_{\sigma \in \tilde{\Sigma}}\},q_{1},Q_{a},Q_{r},Q_{\circlearrowleft} ) $,
	can be exactly simulated by RT-QFA$ ^{\circlearrowleft} $
	$ \mathcal{M}' = (Q,\Sigma,\{ \mathcal{E}_{\sigma \in \tilde{\Sigma}} \},q_{1},Q_{a},Q_{r}) $,
	where, for each $ \sigma \in \tilde{\Sigma} $, 
	$ \mathcal{E}_{\sigma} = \{E_{\sigma,i} \mid 1 \le i \le 4 \} $ can be defined as follows:
	\begin{itemize}
		\item $ E_{\sigma,1} $ is obtained from $ U_{\sigma} $ by keeping all transitions from 
			the nonhalting states to the others and replacing the others with zeros;
		\item $ E_{\sigma,2} $, $ E_{\sigma,3} $, and $ E_{\sigma,4} $ are zero-one diagonal matrices 
			whose entries are 1 only for the transitions leaving restarting, accepting, and rejecting states, 
			respectively.
	\end{itemize}

The following theorem lets us conclude that the two variants of QFAs with restart are equivalent in language recognition power.

\begin{theorem}
	\label{thm:RT-QFA-restart-simulated-by-RT-KWQFA-restart}
	Any language $ L \subseteq \Sigma^{*} $ recognized by an $ n $-state RT-QFA$ ^{\circlearrowleft} $
	with error bound $ \epsilon $ can be recognized by a $ O(n) $-state RT-KWQFA$ ^{\circlearrowleft} $
	with the same error bound. 
	Moreover, if the expected runtime of the RT-QFA$ ^{\circlearrowleft} $ is $ O(s(|w|)) $,
	then, for a constant $ l>1 $,
	the expected runtime of the RT-KWQFA$  ^{\circlearrowleft} $ is $ O(l^{2|w|}s^{2}(|w|)) $,
	where $ w $ is the input string.
\end{theorem}
\begin{proof}
	See Appendix \ref{app:proof-of-QFA-to-KWQFA}.
\end{proof}
% SSSSSSSSSSSSSSSSSSSSSSSSSSSSSSSSSSSSSSSSSSSSSSSSSSSSSSSSSSSSSSSSSSSSSSSSSSSSSSSS %
% SSSSSSSSSSSSSSSSSSSSSSSSSSSSSSSSSSSSSSSSSSSSSSSSSSSSSSSSSSSSSSSSSSSSSSSSSSSSSSSS %
% SSSSSSSSSSSSSSSSSSSSSSSSSSSSSSSSSSSSSSSSSSSSSSSSSSSSSSSSSSSSSSSSSSSSSSSSSSSSSSSS %
\section{Postselection} \label{section:Postselection}
% SSSSSSSSSSSSSSSSSSSSSSSSSSSSSSSSSSSSSSSSSSSSSSSSSSSSSSSSSSSSSSSSSSSSSSSSSSSSSSSS %
% SSSSSSSSSSSSSSSSSSSSSSSSSSSSSSSSSSSSSSSSSSSSSSSSSSSSSSSSSSSSSSSSSSSSSSSSSSSSSSSS %
% SSSSSSSSSSSSSSSSSSSSSSSSSSSSSSSSSSSSSSSSSSSSSSSSSSSSSSSSSSSSSSSSSSSSSSSSSSSSSSSS %

We are now ready to present our model of the real-time finite automaton with
postselection (RT-PostFA).

% sssssssssssssssssssssssssssssssssssssssssssssssssssssssssssssssssssssssssssssssss %
% sssssssssssssssssssssssssssssssssssssssssssssssssssssssssssssssssssssssssssssssss %
\subsection{Definitions} \label{sec:Posdefs}
% sssssssssssssssssssssssssssssssssssssssssssssssssssssssssssssssssssssssssssssssss %
% sssssssssssssssssssssssssssssssssssssssssssssssssssssssssssssssssssssssssssssssss %

A RT-PFA with postselection (RT-PostPFA) is a 5-tuple
\begin{equation}
       \mathcal{P}=(Q,\Sigma,\{ A_{\sigma \in \tilde{\Sigma}} \},q_{1},Q_{p}),
\end{equation}
where $ Q_{p} \subseteq Q$, the only item in this definition that differs from that of the standard RT-PFA, is the set of \textit{postselection states}. $ Q_{p}  $ is  the union of two disjoint sets
$ Q_{pa} $ and $ Q_{pr} $, which are called the postselection accept and reject states, respectively. The remaining states in $Q$ form the set of \textit{nonpostselection states}.

A RT-PostPFA can be seen as a standard RT-PFA satisfying the condition that for each input string $ w \in \Sigma^{*} $,
\begin{equation}
       \sum_{q_{i} \in Q_{p}} v_{|\tilde{w}|}[i] > 0.
\end{equation}

The \textit{acceptance and rejection probabilities of input string $w$ by RT-PostPFA $ \mathcal{P} $  before postselection} are defined as
\begin{equation}
       p_{\mathcal{P}}^{a}(w) = \sum_{q_{i} \in Q_{pa}} v_{|\tilde{w}|}[i]
\end{equation}
and
\begin{equation}
       p_{\mathcal{P}}^{r}(w) = \sum_{q_{i} \in Q_{pr}} v_{|\tilde{w}|}[i].
\end{equation}
Note that we are using notation identical to that introduced in the discussion for automata with restart for these probabilities; the reason will be evident shortly.

Finite automata with postselection have the capability of discarding all computation branches except
the ones belonging to $ Q_{p} $ when they arrive at the end of the input. The probabilities that RT-PostPFA $ \mathcal{P} $ will accept or reject string $w$ are obtained by normalization, and are given by
\begin{equation}
\label{eq:postacc}
       f_{\mathcal{P}}^{a}(w) =
\dfrac{p_{\mathcal{P}}^{a}(w)}{p_{\mathcal{P}}^{a}(w)+p_{\mathcal{P}}^{r}(w)},
\end{equation}
and
\begin{equation}
\label{eq:postrej}
       f_{\mathcal{P}}^{r}(w) =
\dfrac{p_{\mathcal{P}}^{r}(w)}{p_{\mathcal{P}}^{a}(w)+p_{\mathcal{P}}^{r}(w)}.
\end{equation}

The class of languages recognized by RT-PostPFAs with bounded error will be denoted $ \mathsf{PostS} $. The subset of $ \mathsf{PostS} $ consisting of languages recognized by RT-PostPFAs with zero error is denoted $ \mathsf{ZPostS} $.

Quantum finite automata with postselection are defined in a manner completely analogous to their classical counterparts, and are based on the RT-QFA model of Section \ref{sec:QFA}. A RT-QFA with postselection (RT-PostQFA) is a 5-tuple
\begin{equation}
       \mathcal{M}=(Q,\Sigma,\{\mathcal{E}_{\sigma \in \tilde{\Sigma}}\},q_{1},Q_{p}),
\end{equation}
satisfying the condition that for each input string $ w \in \Sigma^{*} $,
\begin{equation}
       tr(P_{p} \rho_{| \tilde{w} |} ) > 0,
\end{equation}
where $ P_{p} $ is the projector defined as
\begin{equation}
       P_{p} = \sum_{q \in Q_{p}} \ket{q}\bra{q}.
\end{equation}
Additionally we define the projectors 
\begin{equation}
       P_{pa} = \sum_{q \in Q_{pa}} \ket{q}\bra{q}
\end{equation}
and
\begin{equation}
       P_{pr} = \sum_{q \in Q_{pr}} \ket{q}\bra{q}.
\end{equation}

The \textit{acceptance and rejection probabilities of input string $w$ by RT-PostQFA $ \mathcal{M} $  before postselection} are defined as
\begin{equation}
       p_{\mathcal{M}}^{a}(w) = tr(P_{pa} \rho_{| \tilde{w} |} )
\end{equation}
and
\begin{equation}
       p_{\mathcal{M}}^{r}(w) = tr(P_{pr} \rho_{| \tilde{w} |} ).
\end{equation}
The probabilities $ f_{\mathcal{M}}^{a}(w) $, $ f_{\mathcal{M}}^{r}(w) $ associated by RT-PostQFAs
are defined similarly to those of RT-PostPFAs.

The class of languages recognized by RT-PostQFAs with bounded error is denoted $ \mathsf{PostQAL} $. The subset of $ \mathsf{PostQAL} $ consisting of languages recognized by RT-PostQFAs with zero error is named $ \mathsf{ZPostQAL} $.
% sssssssssssssssssssssssssssssssssssssssssssssssssssssssssssssssssssssssssssssssss %
% sssssssssssssssssssssssssssssssssssssssssssssssssssssssssssssssssssssssssssssssss %
\subsection{The power of postselection}
% sssssssssssssssssssssssssssssssssssssssssssssssssssssssssssssssssssssssssssssssss %
% sssssssssssssssssssssssssssssssssssssssssssssssssssssssssssssssssssssssssssssssss %

It is evident from the similarity of the statement of Lemma \ref{lem:overall-acc-rej} and Equations \ref{eq:postacc} and \ref{eq:postrej} that there is a close relationship between machines with restart and postselection automata. This is set out in the following theorem.

\begin{theorem}
\label{thm:posres}
	The classes of languages recognized by RT-PFA$ ^{\circlearrowleft} $ and RT-QFA$ ^{\circlearrowleft} $
	with bounded error are identical to $ \mathsf{PostS} $ and $ \mathsf{PostQAL} $, respectively.
\end{theorem}
\begin{proof}
	Given a (probabilistic or quantum) RT-PostFA $\mathcal{P}$, we can construct a corresponding machine with restart $\mathcal{M}$ whose accept and reject states are $\mathcal{P}$'s postselection accept and reject states, respectively. All remaining states of $\mathcal{P}$ are designated as restart states of $\mathcal{M}$.
	Given a machine with restart  $\mathcal{M}$, (we assume the computation is restarted and halted
	only when the input head is at the right end of the tape,) we construct a corresponding RT-PostFA  $\mathcal{P}$ by designating the accept and reject states of $\mathcal{M}$ as  the postselection accept and reject states of $\mathcal{P}$, respectively, and the remaining states of $\mathcal{M}$ are converted to be nonpostselection states. 
	
	By Lemma \ref{lem:overall-acc-rej} and Equations \ref{eq:postacc} and \ref{eq:postrej}, the machines before and after these conversions recognize the same language, with the same error bound.	
\end{proof}

\begin{corollary}
	$ \mathsf{PostQAL} $ and $ \mathsf{PostS} $ are subsets of the classes of the languages recognized
	by two-way QFAs and PFAs, respectively, with bounded error.
\end{corollary}

We are now able to demonstrate that postselection increases the recognition power of both probabilistic and quantum real-time machines. It is well known that finite automata of these types with two-way access to their tape are more powerful than their real-time versions. We do not know if machines with restart equal general two-way automata in power, but we do know that they recognize certain nonregular languages. For a given string $ w $, let $ |w|_{\sigma} $ denote the number of occurrences of symbol $ \sigma $ in $ w $.

\begin{corollary}
	\label{corollary:L-eq}
	$\mathsf{REG} \subsetneq \mathsf{PostS}$.
\end{corollary}
\begin{proof}
	The nonregular language $ L_{eq} = \{ w \in \{a,b\}^{*} \mid |w|_{a} = |w|_{b} \} $ can be recognized by 
	a RT-PFA$ ^{\circlearrowleft} $ \cite{YS10B}. 
\end{proof}

We also show that quantum postselection machines outperform their classical counterparts:

\begin{corollary}
	\label{corollary:L-pal}
		$ \mathsf{PostS} \subsetneq \mathsf{PostQAL} $.
\end{corollary}
\begin{proof}
	$ L_{pal} = \{w \in \{a,b\}^{*} \mid w = w ^{r} \} $ is in $ \mathsf{PostQAL}$, since there exists 
	a RT-QFA$ ^{\circlearrowleft} $ algorithm for recognizing it \cite{YS10B}.
	However,  $ L_{pal} $ cannot be recognized with bounded error even by two-way PFAs \cite{DS92}.
\end{proof}

The recognition error of a given real-time machine with postselection can be
reduced to any desired positive value by performing a tensor product of the machine with itself, essentially running as many parallel copies of it as required.
Specifically, if we combine $ k $ copies of a machine with postselection state set
$ Q_{pa} \cup Q_{pr} $, the new postselection accept and reject 
state sets can be chosen as
\begin{equation}
       Q_{pa}^{\prime} = \underbrace{Q_{pa} \times \cdots \times Q_{pa}}_{k
\mbox{ times}}
\end{equation}
and
\begin{equation}
       Q_{pr}^{\prime} = \underbrace{Q_{pr} \times \cdots \times Q_{pr}}_{k
\mbox{ times}},
\end{equation}
respectively. Note that the postselection feature enables this technique to be simpler than the usual ``majority vote" approach for probability amplification. This is easy to see for probabilistic machines. See Appendix \ref{app:probamp} for a proof for the quantum version.

\begin{theorem}
	\label{thm:post-closure}
	$ \mathsf{PostQAL} $ and $ \mathsf{PostS} $ are closed under complementation, union, and intersection.
\end{theorem}
\begin{proof}
	For any language recognized by a RT-PostFA with bounded error, we can obtain a new RT-PostFA
	recognizing the complement of that language with bounded error, by just swapping the designations of the postselection accept and reject states.
	Therefore, both classes are closed under complementation.
	
	Let $ L_{1} $ and $ L_{2} $ be members of $ \mathsf{PostQAL} $ (resp., $ \mathsf{PostS} $).
	Then, there exist  two RT-PostQFAs (resp., RT-PostPFAs) $ \mathcal{P}_{1} $ and $ \mathcal{P}_{2} $	
	recognizing $ L_{1} $ and $ L_{2} $ with error bound $ \epsilon \leq \frac{1}{4} $, respectively.
	Moreover, let $ Q_{pa_{1}} $ and $ Q_{pr_{1}} $ (resp., $ Q_{pa_{2}} $ and $ Q_{pr_{2}} $)
	represent the sets of postselection accept and reject states of $ \mathcal{P}_{1} $
	(resp., $ \mathcal{P}_{2} $), respectively, and let $ Q_{p_{1}} =  Q_{pa_{1}} \cup Q_{pr_{1}} $
	and $ Q_{p_{2}} =  Q_{pa_{2}} \cup Q_{pr_{2}} $.
	By taking the tensor products of  $ \mathcal{P}_{1} $ and $ \mathcal{P}_{2} $, we obtain two new machines, 
	say $ \mathcal{M}_{1} $ and $ \mathcal{M}_{2} $, and set their definitions so that
	\begin{itemize}
		\item the sets of the postselection accept and reject  states of $ \mathcal{M}_{1} $ are
			\begin{equation}
				Q_{p_{1}} \otimes Q_{p_{2}} \setminus Q_{pr_{1}} \otimes Q_{pr_{2}}
			\end{equation} 
			and
			\begin{equation}
				Q_{pr_{1}} \otimes Q_{pr_{2}},
			\end{equation}
			respectively, and
		\item the sets of the postselection accept and reject  states of $ \mathcal{M}_{2} $ are
			\begin{equation}
				Q_{pa_{1}} \otimes Q_{pa_{2}},
			\end{equation} 
			and
			\begin{equation}
				Q_{p_{1}} \otimes Q_{p_{2}} \setminus Q_{pa_{1}} \otimes Q_{pa_{2}},
			\end{equation}
			respectively.
	\end{itemize}
	Thus, the following inequalities can be verified for a given input string $ w \in \Sigma^{*} $:
	\begin{itemize}
		\item if $ w \in L_{1} \cup L_{2} $, $ f_{\mathcal{M}_{1}}^{a}(w) \ge \frac{15}{16} $;
		\item if $ w \notin L_{1} \cup L_{2} $, $ f_{\mathcal{M}_{1}}^{a}(w) \leq \frac{7}{16} $;
		\item if $ w \in L_{1} \cap L_{2} $, $ f_{\mathcal{M}_{2}}^{a}(w) \ge \frac{9}{16} $; 
		\item if $ w \notin L_{1} \cap L_{2} $, $ f_{\mathcal{M}_{2}}^{a}(w) \leq \frac{1}{16} $.
	\end{itemize}
	We conclude that both classes are closed under union and intersection.
\end{proof}

\begin{theorem}
	\label{thm:PostQ-subset-Q}
	$ \mathsf{PostQAL} $ and $ \mathsf{PostS} $ are subsets of $ \mathsf{S} $ $ ( \mathsf{QAL} ) $.
\end{theorem}
\begin{proof}
	A given RT-PostFA can be converted to its counterpart in the corresponding standard model (without
	 postselection) as follows:
	\begin{itemize}
		\item All nonpostselection states of the RT-PostFA are made to transition to accept states
			with probability $ \frac{1}{2} $ at the end of the computation.
		\item All members of $ Q_{pa} $ are accept states in the new machine.
	\end{itemize}
	Therefore, strings which are members of the original machine's language are accepted with probability 
	exceeding $ \frac{1}{2} $ by the new machine. For other strings, the acceptance probability can be at most $ \frac{1}{2} $.
\end{proof}

By using the fact that $ \mathsf{S} $ is not closed under union and intersection
\cite{Fl72,Fl74,La74}, Corollary \ref{corollary:L-pal}, and Theorems
\ref{thm:post-closure} and \ref{thm:PostQ-subset-Q},
we obtain the following corollary.

\begin{corollary}
       $ \mathsf{PostS} $ $ \subsetneq $ $ \mathsf{PostQAL} $ $ \subsetneq $ $ \mathsf{S} $ $ ( \mathsf{QAL} ) $.
\end{corollary}

For instance, for any triple of integers $u$, $v$, $w$, where $0<u<v<w$, the languages $L_{1}= \{a^{m}b^{k}c^{n} | m^{u}>k^{v}>0\}$ and $L_{2}= \{a^{m}b^{k}c^{n} | k^{v}>n^{w}>0\}$ are in $\mathsf{S}$, 
whereas $L_{1} \cup L_{2}$ is not \cite{Tu82}. It must therefore be the case that at least one of $L_{1}$ and $L_{2}$ is not in $\mathsf{PostQAL}$.

Let us examine the extreme case where we wish our machines to make no error at all. Consider a  RT-PostPFA (or RT-PostQFA) $\mathcal{M}$ that recognizes a language $L$ with zero error. It is easy to see that we can convert $\mathcal{M}$ to a standard RT-PFA (or RT-QFA) $\mathcal{M'}$ that recognizes $L$ with cutpoint zero by just designating the postselection accept states of $\mathcal{M}$ as the accept states of $\mathcal{M'}$. We can build another RT-PFA (or RT-QFA) $\mathcal{M''}$ that recognizes the complement of $L$ with cutpoint zero by designating only the postselection reject states of $\mathcal{M}$ as the accept states of $\mathcal{M''}$. We therefore have the following:

\begin{corollary}
	$ \mathsf{REG} $=$ \mathsf{ZPostS} $ $ \subseteq $ $ \mathsf{ZPostQAL} $ $ \subseteq $ 
	$ \mathsf{NQAL} $ $ \cap $ $ \mathsf{coNQAL} $.
\end{corollary}

(Note that it is still open \cite{YS10A} whether $ \mathsf{NQAL} $ $ \cap $ $ \mathsf{coNQAL} $ 
contains a nonregular language or not.)
We can conclude, using Corollaries \ref{corollary:L-eq} and \ref{corollary:L-pal}, and the fact that
$ L_{pal} $ is not in $ \mathsf{NQAL} $ $ \cap $ $ \mathsf{coNQAL} $ \cite{YS10A}, 
that allowing postselection machines to commit a small but nonzero amount of error increases their recognition power:

\begin{corollary}
	$ \mathsf{ZPostS} $ $ \subsetneq $ $ \mathsf{PostS} $, and $ \mathsf{ZPostQAL} $ $ \subsetneq $ $ \mathsf{PostQAL} $.
\end{corollary}

% SSSSSSSSSSSSSSSSSSSSSSSSSSSSSSSSSSSSSSSSSSSSSSSSSSSSSSSSSSSSSSSSSSSSSSSSSSSSSSSS %
% SSSSSSSSSSSSSSSSSSSSSSSSSSSSSSSSSSSSSSSSSSSSSSSSSSSSSSSSSSSSSSSSSSSSSSSSSSSSSSSS %
% SSSSSSSSSSSSSSSSSSSSSSSSSSSSSSSSSSSSSSSSSSSSSSSSSSSSSSSSSSSSSSSSSSSSSSSSSSSSSSSS %
\section{Latvian PostFAs } \label{section:LPostFA}
% SSSSSSSSSSSSSSSSSSSSSSSSSSSSSSSSSSSSSSSSSSSSSSSSSSSSSSSSSSSSSSSSSSSSSSSSSSSSSSSS %
% SSSSSSSSSSSSSSSSSSSSSSSSSSSSSSSSSSSSSSSSSSSSSSSSSSSSSSSSSSSSSSSSSSSSSSSSSSSSSSSS %
% SSSSSSSSSSSSSSSSSSSSSSSSSSSSSSSSSSSSSSSSSSSSSSSSSSSSSSSSSSSSSSSSSSSSSSSSSSSSSSSS %

The first examination of QFAs with postselection was carried out in \cite{LSF09} by L\={a}ce, Scegulnaja-Dubrovska and Freivalds. The main difference between their machines, which we call \textit{Latvian RT-PostFAs}, and abbreviate as RT-LPostFAs, and ours is that the transitions of RT-LPostFAs are not assumed to lead the machine to at least one postselection state with nonzero probability. RT-LPostFAs have the additional unrealistic capability of detecting if the total probability of postselection states is zero at the end of the processing of the input, and accumulating all probability in a single output in such a case.

Although the motivation for this feature is not explained in
\cite{LSF09,SLF10}, such an approach may be seen as an attempt to compensate for some fundamental weaknesses of finite automata.
In many computational models with bigger space bounds, 
one can modify a machine employing the Latvian approach without changing the recognized language 
so that the postselection state set will have nonzero probability for any input string. 
This is achieved by just creating some
computational paths that end up in the postselection set with
sufficiently small probabilities so that
their inclusion does not change the acceptance probabilities of strings that lead the original machine to the postselection set significantly. These paths can be used to accept or to reject
the input as desired whenever there is zero probability of observing the other postselection states.
Unfortunately, we do not know how to implement this construction
in quantum or probabilistic finite automata\footnote{In a similar vein,
we do not know how to modify a given quantum or probabilistic automaton without changing the recognized language so that it is guaranteed to accept all strings with probabilities that are not equal to the cutpoint.
A related open problem is whether $ \mathsf{coS} = \mathsf{S} $ or not \cite{Pa71,YS10A},
even when we restrict ourselves to computable transition probabilities \cite{Di77}.}, so we prefer our model, in which the only nonstandard capability conferred to the machines is postselection, to the Latvian one.

We will consider LPostQFAs\footnote{The original definitions of RT-LPostQFAs in \cite{LSF09} are based on weaker QFA variants, including the KWQFA. Replacing those with the machines of Section \ref{sec:QFA} does not change the model, by an argument that is almost identical to the one presented in Appendix \ref{app:proof-of-QFA-to-KWQFA}. Only quantum machines are defined in \cite{LSF09}; the probabilistic variant is considered for the first time in this paper.} as machines of the type introduced in Section \ref{sec:Posdefs} with an additional component $ \tau \in \{A,R\} $, such that
whenever the postselection probability is zero for a given input
string $ w \in \Sigma^{*} $,
\begin{itemize}
       \item $ w $ is accepted with probability 1 if $ \tau = A $,
       \item $ w $ is rejected with probability 1 if $ \tau = R $.
\end{itemize}
The bounded-error (resp., zero-error) classes corresponding to the RT-LPostPFA and RT-LPostQFA
models are called $ \mathsf{LPostS} $ and $ \mathsf{LPostQAL} $ (resp., $ \mathsf{ZLPostS} $ 
and $ \mathsf{ZLPostQAL} $), respectively.

\begin{theorem}
	$ \mathsf{LPostS} $ = $ \mathsf{PostS} $.
\end{theorem}
\begin{proof}
	We only need to show that $ \mathsf{LPostS} $ $ \subseteq $ $ \mathsf{PostS} $, the other direction is trivial.
	Let $ L $ be in $ \mathsf{LPostS} $ and let $ \mathcal{P} $ with state set $Q$, postselection states $Q_{p}= Q_{pa} \cup  Q_{pr} $, and $ \tau \in \{A,R\} $ be the RT-LPostPFA
	recognizing $ L $ with error bound $ \epsilon < \frac{1}{2} $.
	Suppose that $ L' $ is the set of strings that lead 
	$ \mathcal{P} $ to the postselection set with zero probability.
	By designating all postselection states as accepting states and removing the probability
	values of transitions, we obtain a real-time nondeterministic finite automaton 
	which recognizes $ \overline{L'} $.
	Thus, there exists a real-time deterministic finite automaton, say $ \mathcal{D} $, recognizing $ L' $. Let $Q_{ \mathcal{D} }$ and $A_{\mathcal{D} }$ be the overall state set, and the set of accept states of $ \mathcal{D} $, respectively.

We combine $ \mathcal{P} $ and $ \mathcal{D} $ with a tensor product to obtain a RT-PostPFA $ \mathcal{P}' $. The postselection state set of $ \mathcal{P}' $ is 
$((Q \setminus Q_{p})\otimes A_{\mathcal{D}}) \cup (Q_{p}\otimes(Q_{ \mathcal{D}} \setminus A_{\mathcal{D}}))$. 
The postselection accept states of $ \mathcal{P}' $ are:
\begin{equation}
	\left\lbrace
	\begin{array}{lll}
		((Q \setminus Q_{p})\otimes A_{\mathcal{D}}) \cup (Q_{pa}\otimes(Q_{ \mathcal{D}} \setminus A_{\mathcal{D}}))
			& , ~~ & \tau = ``A" \\
		Q_{pa}\otimes(Q_{ \mathcal{D}} \setminus A_{\mathcal{D}}) & , & \tau = ``R"
	\end{array}
	\right..
\end{equation}
	$ \mathcal{P}' $ is structured so that 
	if the input string $w$ is in $ L' $, 
	the decision is given deterministically with respect to $ \tau $, and
	if  $w \notin L' $, (that is, the probability of postselection by $ \mathcal{P} $ is nonzero,)
	the decision is given by the standard postselection procedure.
	Therefore, $ L $ is recognized by $ \mathcal{P}' $ with the same error bound as  $ \mathcal{P} $, meaning that $ L \in \mathsf{PostS}$.
\end{proof}

\begin{corollary}
	$ \mathsf{ZPostS} $=$ \mathsf{ZLPostS} $.
\end{corollary}

However, we cannot use the same idea in the quantum case due to the
fact that the class of the languages recognized
by real-time quantum finite automata with cutpoint zero ($ \mathsf{NQAL} $)
is a proper superclass of $ \mathsf{REG} $ \cite{BP02,NIHK02,YS10A}.

\begin{lemma}
	$ \mathsf{NQAL} $ $ \cup $ $ \mathsf{coNQAL} $ $ \subseteq $ $ \mathsf{ZLPostQAL} $.
\end{lemma}
\begin{proof}
	For $ L \in $ $ \mathsf{NQAL} $, designate the accepting states of the QFA recognizing $ L $ with cutpoint zero as postselection accepting states
	with $ \tau = R $. (There are no postselection reject states.)
	For $ L \in $ $ \mathsf{coNQAL} $, designate the accepting states of the QFA recognizing $ \overline{L} $ with cutpoint zero
	as postselection rejecting states with $ \tau = A $. (There are no postselection accept states.)
\end{proof}
\begin{lemma}
	$ \mathsf{ZLPostQAL} $ $ \subseteq $ $ \mathsf{NQAL} $ $ \cup $ $ \mathsf{coNQAL} $.
\end{lemma}
\begin{proof}
	Let $ L $ be a member of $ \mathsf{ZLPostQAL} $ and $ \mathcal{M} $ be a RT-LPostQFA recognizing $ L $ with zero error.
	If $ \tau=R $, for all $ w \in L $, we have that $ p^{a}_{\mathcal{M}}(w) $ is nonzero, and $ p^{r}_{\mathcal{M}}(w) = 0 $.
	Thus, we can design a RT-QFA recognizing $ L $ with cutpoint zero.
	Similarly, if $ \tau=A $, 
	for all $ w \notin L $, we have $ p^{r}_{\mathcal{M}}(w) $ is nonzero, and $ p^{a}_{\mathcal{M}}(w) = 0 $.
	Thus, we can design a RT-QFA recognizing $ \overline{L} $ with cutpoint zero.
\end{proof}
\begin{theorem}
	$ \mathsf{ZLPostQAL} $ $ = $ $ \mathsf{NQAL} $ $ \cup $ $ \mathsf{coNQAL} $.
\end{theorem}

\begin{lemma}
	$ L_{eq\overline{eq}} = \{ aw_{1} \cup bw_{2} \mid w_{1} \in L_{eq}, w_{2} \in \overline{L_{eq}} \}  \in $
	$ \mathsf{PostS} $.
\end{lemma}
\begin{proof}
	Since $ L_{eq} $ is a member of $ \mathsf{PostS} $, $ \overline{L_{eq}} $ is also a member of $ \mathsf{PostS} $.
	Therefore, it is not hard to design a RT-PostPFA recognizing $ L_{eq\overline{eq}} $.
\end{proof}
Since $ L_{eq\overline{eq}} $ is not a member of $ \mathsf{NQAL} $ $ \cup $ $ \mathsf{coNQAL} $ \cite{YS10A},
we can obtain the following theorem.
\begin{theorem}
	$ \mathsf{ZLPostQAL} $ $ \subsetneq $ $ \mathsf{LPostQAL} $.
\end{theorem}

By using the fact\footnote{$ L_{pal} $ was proven to be in $\mathsf{ZLPostQAL}$ for the first time in \cite{LSF09}.} that $ L_{pal} \in $ $ \mathsf{coNQAL} $ $ \setminus $ $ \mathsf{NQAL} $ \cite{YS10A},
we can state that RT-LPostQFAs are strictly more powerful than RT-PostQFAs, at least in the zero-error mode:

\begin{corollary}
	$ \mathsf{ZPostQAL} $ $ \subsetneq $ $ \mathsf{ZLPostQAL} $.
\end{corollary}

\begin{theorem}
	$ \mathsf{LPostQAL} $ is closed under complementation.
\end{theorem}
\begin{proof}
	If a language is recognized by a RT-LPostQFA with bounded error,
	by swapping the accepting and rejecting postselection states and
	by setting $ \tau $ to $ \{A,R\} \setminus \tau $, we obtain a new RT-LPostQFA
	recognizing the complement of the language with bounded error.
	Therefore, $ \mathsf{LPostQAL} $ is closed under complementation.
\end{proof}

\begin{theorem}
	\label{thm:LPostQAL-subset-uQAL}
	$ \mathsf{LPostQAL} $ $ \subseteq $ $ \mathsf{uQAL} $ $ (\mathsf{uS}) $.
\end{theorem}
\begin{proof}
	The proof is similar to the proof of Theorem \ref{thm:PostQ-subset-Q} with the exception that
	\begin{itemize}
		\item if $ \tau = A $, we have recognition with nonstrict cutpoint;
		\item if $ \tau = R $, we have recognition with strict cutpoint.
	\end{itemize}
\end{proof}

It was shown in \cite{DF10} that the language $ L_{say} $, i.e.
\begin{equation}
	\{ w  \mid \exists u_{1},u_{2},v_{1},v_{2} \in \{a,b\}^{*},
			 w = u_{1}bu_{2} = v_{1}bv_{2}, |u_{1}| = |v_{2}| \},
\end{equation}
cannot be recognized by a RT-LPostQFA. 
Since $ L_{say} \notin \mathsf{uS} $ \cite{FYS10A}, the same result follows easily from Theorem 
\ref{thm:LPostQAL-subset-uQAL}.

%As a special remark, by using our results, 
%we can easily follow the previously presented results in the literature, i.e. 
%(i) in \cite{SLF10,SLF10}, it was shown that $ L_{pal} $ is a member of LPostQFA
%(ii) in \cite{DF10}, it was shown that language $ L_{say} = \{ ubvby \mid u,v,y \in \{a,b\}^{*}, |u|=|v| \} $,
%which is not a member of $ \mathsf{uQAL} $ ($ \mathsf{uS} $) \cite{YS09A,YS09C,FYS10A,YS10C},
%cannot be a member of LPostQFA.

% SSSSSSSSSSSSSSSSSSSSSSSSSSSSSSSSSSSSSSSSSSSSSSSSSSSSSSSSSSSSSSSSSSSSSSSSSSSSSSSS %
% SSSSSSSSSSSSSSSSSSSSSSSSSSSSSSSSSSSSSSSSSSSSSSSSSSSSSSSSSSSSSSSSSSSSSSSSSSSSSSSS %
% SSSSSSSSSSSSSSSSSSSSSSSSSSSSSSSSSSSSSSSSSSSSSSSSSSSSSSSSSSSSSSSSSSSSSSSSSSSSSSSS %
\section{Bigger space complexity classes with postselection} \label{section:TMS}
% SSSSSSSSSSSSSSSSSSSSSSSSSSSSSSSSSSSSSSSSSSSSSSSSSSSSSSSSSSSSSSSSSSSSSSSSSSSSSSSS %
% SSSSSSSSSSSSSSSSSSSSSSSSSSSSSSSSSSSSSSSSSSSSSSSSSSSSSSSSSSSSSSSSSSSSSSSSSSSSSSSS %
% SSSSSSSSSSSSSSSSSSSSSSSSSSSSSSSSSSSSSSSSSSSSSSSSSSSSSSSSSSSSSSSSSSSSSSSSSSSSSSSS %
Let us briefly discuss how our results can be generalized to probabilistic or quantum Turing machines with nonconstant space bounds. For that purpose, we start by imposing the standard convention that all transition probabilities or amplitudes used in Turing machines with postselection should be restricted to efficiently computable real numbers. It is evident that the relationship between the classes of languages recognized by real-time machines with postselection and real-time machines  
with restart established in Theorem \ref{thm:posres} can be generalized for all space bounds, since one does not consume any additional space when one resets the head, and switches to the initial state.

Machines with postselection with two-way input tape heads need not halt at the end of the first pass of the input, and therefore have to be defined in a slightly different manner, such that the overall state set is partitioned into four subsets, namely, the sets of postselection accept, postselection reject, nonpostselection halting, and nonpostselection non-halting states. Since two-way machines are already able to implement restarts, adding that capability to them
does not increase the language recognition power.
Therefore, for any space bound, the class of languages recognized with bounded error by a two-way machine with postselection equals the class recognized by the standard version of that machine, which also forms a natural bound
for the power real-time or one-way versions of that model with postselection.

We also note that the language $L_{pal}$ cannot be recognized with bounded error by probabilistic Turing machines for any space bound $o(\log n)$ \cite{FK94}, and Corollary \ref{corollary:L-pal} can therefore be generalized to state that (two-way/one-way/real-time) quantum Turing machines with postselection are superior in recognition power to their probabilistic counterparts for any such bound.

%Note that, the definition of a two-way machine with postselection can be defined 
%as making a selection on the classical outcomes of the computation.
%Therefore, more than one accepting and rejecting case should be defined for a quantum machine.

% SSSSSSSSSSSSSSSSSSSSSSSSSSSSSSSSSSSSSSSSSSSSSSSSSSSSSSSSSSSSSSSSSSSSSSSSSSSSSSSS %
% SSSSSSSSSSSSSSSSSSSSSSSSSSSSSSSSSSSSSSSSSSSSSSSSSSSSSSSSSSSSSSSSSSSSSSSSSSSSSSSS %
% SSSSSSSSSSSSSSSSSSSSSSSSSSSSSSSSSSSSSSSSSSSSSSSSSSSSSSSSSSSSSSSSSSSSSSSSSSSSSSSS %
\section{Concluding Remarks} \label{section:ConcludingRemarks}
% SSSSSSSSSSSSSSSSSSSSSSSSSSSSSSSSSSSSSSSSSSSSSSSSSSSSSSSSSSSSSSSSSSSSSSSSSSSSSSSS %
% SSSSSSSSSSSSSSSSSSSSSSSSSSSSSSSSSSSSSSSSSSSSSSSSSSSSSSSSSSSSSSSSSSSSSSSSSSSSSSSS %
% SSSSSSSSSSSSSSSSSSSSSSSSSSSSSSSSSSSSSSSSSSSSSSSSSSSSSSSSSSSSSSSSSSSSSSSSSSSSSSSS %

\begin{figure}[t!]
	\begin{center}
	\fbox{
	\begin{minipage}{0.9\textwidth}
		\begin{center}
		~~\\~~\\
		\ifx\JPicScale\undefined\def\JPicScale{1}\fi
		\unitlength \JPicScale mm
		\begin{picture}(100,74)(0,0)
			\put(34,0){$ \mathsf{REG} = \mathsf{ZPostS} = \mathsf{ZLPostS} $}
			\put(48,3.5){\vector(-1,1){20}}
			\put(10,25){$ \mathsf{PostS} = \mathsf{LPostS} $}
			\put(43,15){$ \mathsf{ZPostQAL} $}
			\multiput(49,3.5)(0,1){10}{$ . $}
			\put(49.3,13){\vector(0,1){1.5}}
			\put(44,40){$ \mathsf{PostQAL} $}
			\put(28,28){\vector(2,1){20}}
			\put(49,18){\vector(0,1){20}}
			\put(73,35){$ \mathsf{ZLPostQAL} =  $}
			\put(70,30){$ \mathsf{NQAL} \cup \mathsf{coNQAL} $}
			\put(52,18){\vector(2,1){22}}
			\put(72,55){$ \mathsf{LPostQAL} $}
			\multiput(51,43)(2,1){11}{$ . $}
			\put(72,53.5){\vector(2,1){1.5}}
			\put(80,38){\vector(0,1){16}}
			\put(13,45){$ \mathsf{BPSPACE(1)} $}
			\multiput(20,28)(0,1){16}{$ . $}
			\put(20.2,43){\vector(0,1){1.5}}
			\put(13,65){$ \mathsf{BQSPACE(1)} $}
			\put(44,65){$ \mathsf{QAL=S} $}
			\put(70,75){$ \mathsf{uQAL=uS} $}
			\put(20,48){\vector(0,1){16}}
			\put(22,48){\vector(3,2){23}}
			\multiput(48,43)(-1,1){21}{$ . $}
			\put(29.4,62){\vector(-1,1){1.5}}
			\put(49.7,43){\vector(0,1){21}}
			\multiput(55,68)(3,1){6}{$ . $}
			\put(70,73.2){\vector(3,1){1.5}}
			\multiput(78,58)(0,1){16}{$ . $}
			\put(78.2,73){\vector(0,1){1.5}}
		\end{picture}
		\end{center}
	\end{minipage}}
	\end{center}
	\caption{The relationships among classical and quantum constant space-bounded classes}
	\vskip\baselineskip
	\label{fig:conclusion}
\end{figure}
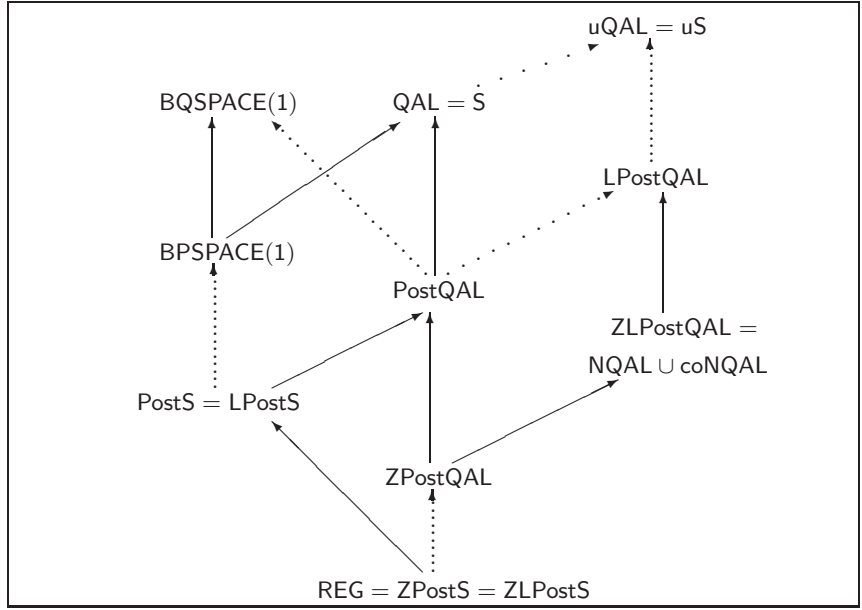 

The relation between postselection and restarting can be extended
easily to other cases.
For example, $\mathsf{PostBQP}$ can also be seen as the class of languages
recognized by
polynomial-size quantum circuits that have been augmented to model the restart action.

Figure \ref{fig:conclusion} summarizes the results presented in this paper. Dotted arrows indicate subset 
relationships, and unbroken arrows represent the cases where it is known that the inclusion is proper. Note that the 
real numbers appearing in the finite automaton definitions are assumed to be restricted as explained in Section 
\ref{section:TMS} for compatibility with the classes based on Turing machines. With that restriction, $\mathsf{BPSPACE(1)}$ and $\mathsf{BQSPACE(1)}$ denote the classes of languages recognized with bounded error by two-way probabilistic and quantum finite automata, respectively.
	
% SSSSSSSSSSSSSSSSSSSSSSSSSSSSSSSSSSSSSSSSSSSSSSSSSSSSSSSSSSSSSSSSSSSSSSSSSSSSSSSS %
% SSSSSSSSSSSSSSSSSSSSSSSSSSSSSSSSSSSSSSSSSSSSSSSSSSSSSSSSSSSSSSSSSSSSSSSSSSSSSSSS %
% SSSSSSSSSSSSSSSSSSSSSSSSSSSSSSSSSSSSSSSSSSSSSSSSSSSSSSSSSSSSSSSSSSSSSSSSSSSSSSSS %
\section*{Acknowledgement} \label{section:Acknowledgement}
% SSSSSSSSSSSSSSSSSSSSSSSSSSSSSSSSSSSSSSSSSSSSSSSSSSSSSSSSSSSSSSSSSSSSSSSSSSSSSSSS %
% SSSSSSSSSSSSSSSSSSSSSSSSSSSSSSSSSSSSSSSSSSSSSSSSSSSSSSSSSSSSSSSSSSSSSSSSSSSSSSSS %
% SSSSSSSSSSSSSSSSSSSSSSSSSSSSSSSSSSSSSSSSSSSSSSSSSSSSSSSSSSSSSSSSSSSSSSSSSSSSSSSS %

We thank R\={u}si\c{n}\v{s} Freivalds for pointing us to the subject of this paper, 
and kindly providing us copies of references \cite{LSF09,SLF10}.

\appendix

% SSSSSSSSSSSSSSSSSSSSSSSSSSSSSSSSSSSSSSSSSSSSSSSSSSSSSSSSSSSSSSSSSSSSSSSSSSSSSSSS %
% SSSSSSSSSSSSSSSSSSSSSSSSSSSSSSSSSSSSSSSSSSSSSSSSSSSSSSSSSSSSSSSSSSSSSSSSSSSSSSSS %
% SSSSSSSSSSSSSSSSSSSSSSSSSSSSSSSSSSSSSSSSSSSSSSSSSSSSSSSSSSSSSSSSSSSSSSSSSSSSSSSS %
\section{The proof of Theorem \ref{thm:RT-QFA-restart-simulated-by-RT-KWQFA-restart}} 
\label{app:proof-of-QFA-to-KWQFA}
% SSSSSSSSSSSSSSSSSSSSSSSSSSSSSSSSSSSSSSSSSSSSSSSSSSSSSSSSSSSSSSSSSSSSSSSSSSSSSSSS %
% SSSSSSSSSSSSSSSSSSSSSSSSSSSSSSSSSSSSSSSSSSSSSSSSSSSSSSSSSSSSSSSSSSSSSSSSSSSSSSSS %
% SSSSSSSSSSSSSSSSSSSSSSSSSSSSSSSSSSSSSSSSSSSSSSSSSSSSSSSSSSSSSSSSSSSSSSSSSSSSSSSS %

We will use almost the same idea presented in the proof of 
Theorem 1 in \cite{YS10B}, in which a similar relationship was shown to hold between the
RT-PFA$ ^{\circlearrowleft} $ and the RT-KWQFA$ ^{\circlearrowleft} $, 
after linearizing the computation of the given RT-QFA$ ^{\circlearrowleft} $.
Let $ \mathcal{G} = (Q,\Sigma,\{ \mathcal{E}_{\sigma \in \tilde{\Sigma}} \},q_{1},Q_{a},Q_{r}) $ 
be an $ n $-state RT-QFA$ ^{\circlearrowleft} $ recognizing $ L $ 	
with error bound $ \epsilon $. We will construct a $ 3n^{2}+6 $-state RT-KWQFA$ ^{\circlearrowleft} $ 
$ \mathcal{M} $ recognizing the same language with error bound $ \epsilon' \le \epsilon $.

The computation of  $ \mathcal{G} $ can be linearized by using techniques described on Page 73 in \cite{Wa03} 
(also see \cite{Ya11A,YS11A}), and so we obtain
a $ n^{2} \times n^{2} $-dimensional matrix for each $ \sigma \in \tilde{\Sigma} $, i.e.
$ A_{\sigma} = \sum\limits_{i=1}^{|\mathcal{E}_{\sigma}|} E_{\sigma,i} \otimes E_{\sigma,i}^{*} $.
We then add two new states, $ q_{n^{2}+1} $ ($ q_{a} $) and $ q_{n^{2}+2} $ ($ q_{r} $), and correspondingly construct new transition matrices so that the overall accepting and rejecting 
probabilities, respectively, are summed up on these new states, i.e. 
\begin{equation}
	A_{\sigma \in \Sigma \cup \{\cent\}}'= \left(
	\begin{array}{c|c}
		A_{\sigma} &  0_{n \times 2}  \\
		\hline
		0_{2 \times n} & I_{2 \times 2} \\
	\end{array}
	\right),
	A_{\dollar}'=\left(
		\begin{array}{c|c}
		0_{n \times n} & 0_{2 \times n}  \\
		\hline
		T_{2 \times n} & I_{2 \times 2} \\
	\end{array}
	\right)
	\left(
	\begin{array}{c|c}
		A_{\dollar} &  0_{n \times 2}  \\
		\hline
		0_{2 \times n} & I_{2 \times 2} \\
	\end{array}
	\right),
\end{equation}
where all the entries of $ T $ are zeros except that 
$ T[1,(i-1)n^{2}+i] = 1  $ when $ q_{i} \in Q_{a} $ and
$ T[2,(i-1)n^{2}+i] = 1  $ when $ q_{i} \in Q_{r} $.
Let $ v_{0} = (1,0,\ldots,0)^{\dagger} $ be a column vector of dimension $ n^{2}+2 $.
It can be verified easily that, for any $ w \in \Sigma^{*} $,
\begin{equation}
	v_{|\tilde{w}|}'=A_{\dollar}' A_{w_{|w|}}' \cdots A_{w_{1}}' A_{\cent}' v_{0}
	= (0_{n^{2} \times 1},p_{\mathcal{G}}^{a}(w),p_{\mathcal{G}}^{r}(w))^{\dagger}.
\end{equation}

\begin{figure}[h!]
	\caption{General template to build an orthonormal set}
	\centering
	\fbox{
	\begin{minipage}{0.9\textwidth}
		\footnotesize
		Let $ S $ be a finite set and $ \{ A_{s} \mid s \in S \} $ 
		be a set of $ m \times m $-dimensional matrices.
		We present a method in order to find two sets of a $ m \times m $-dimensional matrices, 
		$ \{ B_{s} \mid s \in S\} $ and $ \{ C_{s} \mid s \in S\} $,  with a generic constant $ l $ such that
		the columns of the matrix 
		\begin{equation}
			\frac{1}{l} 
			\left( \begin{array}{c} A_{s} \\ \hline B_{s} \\ \hline C_{s} \end{array} \right)
		\end{equation}
form an orthonormal set		for each $ s \in S $. The details are given below.
		\begin{enumerate}
			\item The entries of $ B_{s \in S} $ and $ C_{s \in S} $ are set to 0.
			\item For each $ s \in S $, the entries of $ B_{s} $ are updated to make the columns of
				$ \left( \begin{array}{c} A_{s} \\ \hline B_{s} \end{array} \right) $  pairwise orthogonal.
				Specifically, \\
				\begin{tabular}{ll}
					& for $ i=1, \ldots, m-1 $ \\
					& ~~~~set $ b_{i,i}=1 $ \\
					& ~~~~for $ j=i+1, \ldots, m $ \\
					& ~~~~~~~~set $ b_{i,j} $ to some value so that the $ i^{th} $ and $ j^{th} $ columns 
					become \\ 
					& ~~~~~~~~orthogonal \\
					& set $ l_{s} $ to the maximum of the lengths (norms) of the columns of
					$ \left( \begin{array}{c} A_{s} \\ \hline B_{s} \end{array} \right) $
				\end{tabular}
			\item  Set $ l = \max( \{ l_{s} \mid s \in S \} ) $.
			\item For each $ s \in S $, the diagonal entries of $ C_{s} $ are 
				updated to make the length of each column of
				$ \left( \begin{array}{c} A_{s} \\ \hline B_{s} \\ \hline C_{s} 
				\end{array} \right) $ equal to $ l $.
		\end{enumerate}
	\end{minipage}
	}
	\label{berr:fig:general-template-2}
\end{figure}
	
Based on the template given in Figure \ref{berr:fig:general-template-2}, we calculate a constant $ l $ and
the sets $ B_{\sigma \in \tilde{\Sigma}} $ and $ C_{\sigma \in \tilde{\Sigma}} $,
such that the columns of the matrix
\begin{equation}
	\frac{1}{l} \left( \begin{array}{c} A'_{\sigma} \\ \hline B_{\sigma} 
	\\ \hline C_{\sigma} \end{array} \right)
\end{equation}
form an orthonormal set.
Thus, for each $ \sigma \in \tilde{\Sigma} $, we define transition matrices of $ \mathcal{M} $ as 
\begin{equation}
	U_{\sigma} = \frac{1}{l} \left( \begin{array}{c|c}
	\begin{array}{c} A'_{\sigma} \\ \hline B_{\sigma} 
	\\ \hline C_{\sigma} \end{array} & D_{\sigma}
	\end{array}
	 \right),
\end{equation}
where the entries of $ D_{\sigma} $ are selected in order to make $ U_{\sigma} $ unitary.
The state set of $ \mathcal{M} $ can be specified as follows:
\begin{enumerate}
	\item The states corresponding to $ q_{a} $ and $ q_{r} $ are
	the accepting and rejecting states, respectively,
	\item All the states corresponding to rows (or columns) of the $ A_{\sigma \in \tilde{\Sigma}} $'s 
		are the nonhalting states, where the first one is the initial state, and,
	\item All remaining states are restarting states.
\end{enumerate}

When $ \mathcal{M} $ runs on input string $w$, the amplitudes of $ q_{a}' $ and $ q_{r}' $, 
the only halting states of $ \mathcal{M} $, at the end of the first round are
$ \left( \frac{1}{l} \right)^{|\tilde{w}|}p_{\mathcal{G}}^{a}(w)  $ and
$ \left( \frac{1}{l} \right)^{|\tilde{w}|}p_{\mathcal{G}}^{r}(w)  $, respectively.
We therefore have by Lemma \ref{lem:bounded-error} that when $ w \in L $,
\begin{equation}
	\frac{p_{\mathcal{M}}^{r}(w)}{p_{\mathcal{M}}^{a}(w)} =
	\frac{(p_{\mathcal{G}}^{r}(w))^{2}}{(p_{\mathcal{G}}^{a}(w))^{2}}
	\leq
	\frac{\epsilon^{2}}{(1-\epsilon)^{2}},
\end{equation}
and similarly, when $ w \notin L $,
\begin{equation}
	\frac{p_{\mathcal{M}}^{a}(w)}{p_{\mathcal{M}}^{r}(w)} =
	\frac{(p_{\mathcal{G}}^{a}(w))^{2}}{(p_{\mathcal{G}}^{r}(w))^{2}}
	\leq
	\frac{\epsilon^{2}}{(1-\epsilon)^{2}}.
\end{equation}
By solving the equation
\begin{equation} \frac{\epsilon'}{1-\epsilon'} = \frac{\epsilon^{2}}{(1-\epsilon)^{2}}, \end{equation}
we obtain
\begin{equation} \epsilon'=\frac{\epsilon^{2}}{1 - 2\epsilon + 2\epsilon^{2}} \leq \epsilon. \end{equation}

By using Lemma \ref{lemma:expected-runtime}, the expected runtime of $ \mathcal{G} $ is 
\begin{equation} \frac{1}{p_{\mathcal{G}}^{a}(w)+p_{\mathcal{G}}^{r}(w)} |w| \in O(s(|w|)), \end{equation}
and so the expected runtime of $ \mathcal{M} $ is
\begin{equation}
	\left( l \right)^{2|\tilde{w}|} \frac{1}{(p_{\mathcal{G}}^{a}(w))^{2}+(p_{\mathcal{G}}^{r}(w))^{2}}|w|
	<
	3 \left( l \right)^{2|\tilde{w}|} \left( \frac{1}{p_{\mathcal{G}}^{a}(w)
	+p_{\mathcal{G}}^{r}(w)} \right)^{2} |w|,
\end{equation}
which is $ O(l^{2|w|}s^{2}(|w|))  $.

% SSSSSSSSSSSSSSSSSSSSSSSSSSSSSSSSSSSSSSSSSSSSSSSSSSSSSSSSSSSSSSSSSSSSSSSSSSSSSSSS %
% SSSSSSSSSSSSSSSSSSSSSSSSSSSSSSSSSSSSSSSSSSSSSSSSSSSSSSSSSSSSSSSSSSSSSSSSSSSSSSSS %
% SSSSSSSSSSSSSSSSSSSSSSSSSSSSSSSSSSSSSSSSSSSSSSSSSSSSSSSSSSSSSSSSSSSSSSSSSSSSSSSS %
\section{Error reduction for postselection machines} \label{app:probamp}
% SSSSSSSSSSSSSSSSSSSSSSSSSSSSSSSSSSSSSSSSSSSSSSSSSSSSSSSSSSSSSSSSSSSSSSSSSSSSSSSS %
% SSSSSSSSSSSSSSSSSSSSSSSSSSSSSSSSSSSSSSSSSSSSSSSSSSSSSSSSSSSSSSSSSSSSSSSSSSSSSSSS %
% SSSSSSSSSSSSSSSSSSSSSSSSSSSSSSSSSSSSSSSSSSSSSSSSSSSSSSSSSSSSSSSSSSSSSSSSSSSSSSSS %

\begin{lemma}
	If $ L $ is recognized by RT-PostQFA (resp., RT-PostPFA) $ \mathcal{M} $ with error bound 
	$ \epsilon \in (0,\frac{1}{2})  $,
	then there exists a RT-PostQFA (resp., RT-PostPFA), say $ \mathcal{M}' $,
	recognizing $ L $ with error bound $  \epsilon^{2} $.
\end{lemma}
\begin{proof}
	We give a proof for RT-PostQFAs, which can be adapted easily to RT-PostPFAs.
	$ M' $ can be obtained by taking the tensor product of $ k $ copies of $ \mathcal{M} $, where
	the new postselection accept (resp., reject)  states, $ Q_{pa}' $  (resp., $ Q_{pr}' $),
	are $ \otimes_{i=1}^{k} Q_{pa} $ (resp., $ \otimes_{i=1}^{k} Q_{pr} $),
	where $ Q_{pa} $ (resp., $ Q_{pr} $) are the postselection accept (resp., reject) states of $ \mathcal{M} $.
	
	Let $ \rho_{\tilde{w}} $ and $ \rho_{\tilde{w}}' $ be the  density matrices of 
	$ \mathcal{M} $ and $ \mathcal{M}' $, respectively,
	after reading $ \tilde{w} $ for a given input string $ w \in \Sigma^{*} $. 
	By definition, we have
	\begin{equation}
		p_{\mathcal{M}}^{a}(w) = \sum_{q_{i} \in Q_{pa} }\rho_{\tilde{w}}[i,i],
		~~~~
		p_{\mathcal{M}'}^{a}(w) = \sum_{q_{i'} \in Q_{pa}' }\rho_{\tilde{w}}[i',i']		
	\end{equation}
	and 
	 \begin{equation}
		p_{\mathcal{M}}^{r}(w) = \sum_{q_{i} \in Q_{pr} }\rho_{\tilde{w}}[i,i],
		~~~~
		p_{\mathcal{M}'}^{r}(w) = \sum_{q_{i'} \in Q_{pr}' }\rho_{\tilde{w}}[i',i'].
	\end{equation}
	By using the equality $ \rho_{\tilde{w}}' = \otimes_{i=1}^{k} \rho_{\tilde{w}}  $,
	the following can be obtained with a straightforward calculation:
	\begin{equation}
		p_{\mathcal{M}'}^{a}(w) = \left( p_{\mathcal{M}}^{a}(w) \right)^{k}
	\end{equation}
	and
	\begin{equation}
		p_{\mathcal{M}'}^{r}(w) = \left( p_{\mathcal{M}}^{r}(w) \right)^{k}.
	\end{equation}
	
	We examine the case of $ w \in L $ (the case $ w \notin L $ is symmetric).
	Since $ L $ is recognized by $ \mathcal{M} $ with error bound $ \epsilon $, 
	we have (due to Lemma \ref{lem:bounded-error})
	\begin{equation}
		 \frac{p_{\mathcal{M}}^{r}(w)}{p_{\mathcal{M}}^{a}(w)}
		 \leq 
		 \frac{\epsilon}{1-\epsilon}.
	\end{equation}
	If $ L $ is recognized by $ \mathcal{M}' $ with error bound $ \epsilon^{2} $,
	we must have
	\begin{equation}
		\frac{p_{\mathcal{M}'}^{r}(w)}{p_{\mathcal{M}'}^{a}(w)}
		\leq
		\frac{\epsilon^{2}}{1-\epsilon^{2}}.
	\end{equation}
	Thus, any $ k $ satisfying
	\begin{equation}		
		\left( \frac{\epsilon}{1-\epsilon} \right)^{k}
		\leq		
		\frac{\epsilon^{2}}{1-\epsilon^{2}}
	\end{equation}
provides the desired machine $ \mathcal{M}' $	due to the fact that 
	\begin{equation}
		\label{berr:eq:k}
		\frac{p_{\mathcal{M}'}^{r}(w)}{p_{\mathcal{M}'}^{a}(w)} =
		\left( \frac{p_{\mathcal{M}}^{r}(w)}{p_{\mathcal{M}}^{a}(w)} \right)^{k}.
	\end{equation}
	By solving Equation \ref{berr:eq:k}, we  get
	\begin{equation}
		k = 1 + \left\lceil \frac{ \log \left( \frac{1}{\epsilon} + 1 \right) }{ 
		\log \left( \frac{1}{\epsilon} - 1 \right) } \right\rceil.
	\end{equation}
	Therefore, for any $ 0 < \epsilon < \frac{1}{2} $, we can find a value for $ k $.
\end{proof}

\begin{corollary}
	If $ L $ is recognized by RT-PostQFA (resp., RT-PostPFA) $ \mathcal{M} $ with error bound 
	$ 0 < \epsilon < \frac{1}{2}  $,
	then there exists a RT-PostQFA (resp., RT-PostPFA), say $ \mathcal{M}' $, recognizing $ L $ with error bound 
	$ \epsilon' < \epsilon $ such that $ \epsilon' $ can be arbitrarily close to 0.
\end{corollary}

\bibliographystyle{plain}
\bibliography{YakaryilmazSay}

\begin{thebibliography}{10}

\bibitem{Aa05}
Scott Aaronson.
\newblock Quantum computing, postselection, and probabilistic polynomial-time.
\newblock {\em Proceedings of the Royal Society A}, 461(2063):3473--3482, 2005.

\bibitem{AY11A}
Andris Ambainis and Abuzer Yakary{\i}lmaz.
\newblock {\em Automata: from Mathematics to Applications}, chapter Automata
  and quantum computing.
\newblock (In preparation).

\bibitem{BJKP05}
Vincent~D. Blondel, Emmanuel Jeandel, Pascal Koiran, and Natacha Portier.
\newblock Decidable and undecidable problems about quantum automata.
\newblock {\em SIAM Journal on Computing}, 34(6):1464--1473, 2005.

\bibitem{Bo03}
Symeon Bozapalidis.
\newblock Extending stochastic and quantum functions.
\newblock {\em Theory of Computing Systems}, 36(2):183--197, 2003.

\bibitem{BP02}
Alex Brodsky and Nicholas Pippenger.
\newblock Characterizations of 1--way quantum finite automata.
\newblock {\em SIAM Journal on Computing}, 31(5):1456--1478, 2002.

\bibitem{Bu67}
R.~G. Bukharaev.
\newblock {\em Probabilistic methods and cybernetics. V}, volume 127:3 of {\em
  Gos. Univ. Uchen. Zap.}, chapter On the representability of events in
  probabilistic automata, pages 7--20.
\newblock Kazan, 1967.
\newblock (Russian).

\bibitem{Di77}
Phan~Dinh Di{\^e}u.
\newblock Criteria of representability of languages in probabilistic automata.
\newblock {\em Cybernetics and Systems Analysis}, 13(3):352--364, 1977.
\newblock Translated from Kibernetika, No. 3, pp. 39$ \mbox{--} $50, May$
  \mbox{--} $June, 1977.

\bibitem{DS92}
Cynthia Dwork and Larry Stockmeyer.
\newblock Finite state verifiers $\mbox{I}$: The power of interaction.
\newblock {\em Journal of the ACM}, 39(4):800--828, 1992.

\bibitem{Fl72}
Michel Fliess.
\newblock Automates stochastiques et s{\'e}ries rationnelles non commutatives.
\newblock In {\em Automata, Languages, and Programming}, pages 397--411, 1972.

\bibitem{Fl74}
Michel Fliess.
\newblock Propri\'{e}t\'{e}s bool\'{e}ennes des langages stochastiques.
\newblock {\em Mathematical Systems Theory}, 7(4):353--359, 1974.

\bibitem{Fr81}
R\={u}si\c{n}\v{s} Freivalds.
\newblock Probabilistic two-way machines.
\newblock In {\em Proceedings of the International Symposium on Mathematical
  Foundations of Computer Science}, pages 33--45, 1981.

\bibitem{FK94}
R\={u}si\c{n}\v{s} Freivalds and Marek Karpinski.
\newblock Lower space bounds for randomized computation.
\newblock In {\em ICALP'94: Proceedings of the 21st International Colloquium on
  Automata, Languages and Programming}, pages 580--592, 1994.

\bibitem{FYS10A}
R\={u}si\c{n}\v{s} Freivalds, Abuzer Yakary{\i}lmaz, and A.~C.~Cem Say.
\newblock A new family of nonstochastic languages.
\newblock {\em Information Processing Letters}, 110(10):410--413, 2010.

\bibitem{Hi08}
Mika Hirvensalo.
\newblock Various aspects of finite quantum automata.
\newblock In {\em DLT'08: Proceedings of the 12th international conference on
  Developments in Language Theory}, pages 21--33, 2008.

\bibitem{Je07}
Emmanuel Jeandel.
\newblock Topological automata.
\newblock {\em Theory of Computing Systems}, 40(4):397--407, 2007.

\bibitem{KW97}
Attila Kondacs and John Watrous.
\newblock On the power of quantum finite state automata.
\newblock In {\em FOCS'97: Proceedings of the 38th Annual Symposium on
  Foundations of Computer Science}, pages 66--75, 1997.

\bibitem{LSF09}
Lelde L\={a}ce, Oksana {Scegulnaja-Dubrovska}, and R\={u}si\c{n}\v{s}
  Freivalds.
\newblock Languages recognizable by quantum finite automata with cut-point 0.
\newblock In {\em SOFSEM'09: Proceedings of the 35th International Conference
  on Current Trends in Theory and Practice of Computer Science}, volume~2,
  pages 35--46, 2009.

\bibitem{La74}
J\={a}nis Lapi\c{n}\v{s}.
\newblock On nonstochastic languages obtained as the union and intersection of
  stochastic languages.
\newblock {\em Avtom. Vychisl. Tekh.}, (4):6--13, 1974.
\newblock (Russian).

\bibitem{NIHK02}
Masaki Nakanishi, Takao Indoh, Kiyoharu Hamaguchi, and Toshinobu Kashiwabara.
\newblock On the power of non-deterministic quantum finite automata.
\newblock {\em IEICE Transactions on Information and Systems},
  E85-D(2):327--332, 2002.

\bibitem{Pa71}
Azaria Paz.
\newblock {\em Introduction to Probabilistic Automata}.
\newblock Academic Press, New York, 1971.

\bibitem{Ra63}
Michael~O. Rabin.
\newblock Probabilistic automata.
\newblock {\em Information and Control}, 6:230--243, 1963.

\bibitem{DF10}
Oksana {Scegulnaja-Dubrovska} and R\={u}si\c{n}\v{s} Freivalds.
\newblock A context-free language not recognizable by postselection finite
  quantum automata.
\newblock In R\={u}si\c{n}\v{s} Freivalds, editor, {\em Randomized and quantum
  computation}, pages 35--48, 2010.
\newblock Satellite workshop of MFCS and CSL 2010.

\bibitem{SLF10}
Oksana {Scegulnaja-Dubrovska}, Lelde L\={a}ce, and R\={u}si\c{n}\v{s}
  Freivalds.
\newblock Postselection finite quantum automata.
\newblock volume 6079 of {\em Lecture Notes in Computer Science}, pages
  115--126, 2010.

\bibitem{Tu82}
Paavo Turakainen.
\newblock {\em Discrete Mathematics}, volume~7 of {\em Banach Center
  Publications}, chapter Rational stochastic automata in formal language
  theory, pages 31--44.
\newblock PWN-Polish Scientific Publishers, Warsaw, 1982.

\bibitem{Wa03}
John Watrous.
\newblock On the complexity of simulating space-bounded quantum computations.
\newblock {\em Computational Complexity}, 12(1-2):48--84, 2003.

\bibitem{Ya11A}
Abuzer Yakary{\i}lmaz.
\newblock {\em Classical and Quantum Computation with Small Space Bounds}.
\newblock PhD thesis, Bo\u{g}azi\c{c}i University, 2011.
\newblock (arXiv:1102.0378).

\bibitem{YS10A}
Abuzer Yakary{\i}lmaz and A.~C.~Cem Say.
\newblock Languages recognized by nondeterministic quantum finite automata.
\newblock {\em Quantum Information and Computation}, 10(9\&10):747--770, 2010.

\bibitem{YS10B}
Abuzer Yakary{\i}lmaz and A.~C.~Cem Say.
\newblock Succinctness of two-way probabilistic and quantum finite automata.
\newblock {\em Discrete Mathematics and Theoretical Computer Science},
  12(2):19--40, 2010.

\bibitem{YS11A}
Abuzer Yakary{\i}lmaz and A.~C.~Cem Say.
\newblock Unbounded-error quantum computation with small space bounds.
\newblock {\em Information and Computation}, 2011.
\newblock (To appear) (arXiv:1007.3624).

\end{thebibliography}

\end{document}